\documentclass[a4paper, 10pt, twosides]{amsart}
\usepackage{amssymb}
\usepackage{enumitem}
\setenumerate{nolistsep}
\setitemize{nolistsep}
\usepackage{dsfont}
\usepackage{tikz}
\usetikzlibrary{automata,positioning}
\usepackage{mathtools}
\allowdisplaybreaks

\newtheorem{theorem}{Theorem}
\newtheorem{definition}{Definition}
\newtheorem{fact}{Fact}

\newtheorem{prop}{Proposition}
\newtheorem{remark}{Remark}
\newtheorem{example}{Example}
\newtheorem{problem}{Problem}

\newtheorem{cor}{Corollary}

\newcommand{\norm}[1]{\left\lVert{#1}\right\rVert}
\newcommand{\abs}[1]{\left\lvert{#1}\right\rvert}

\newcommand{\pmat}[1]{\begin{pmatrix}#1\end{pmatrix}}

\renewcommand{\geq}{\geqslant}

\renewcommand{\leq}{\leqslant}

\newcommand{\R}{\mathds{R}}

\renewcommand{\P}{\mathcal{P}}

\newcommand{\T}{\mathrm{T}}
\newcommand{\Ntsigma}{\mathrm{N}}

\newcommand{\Ss}{\mathcal{S}}
\newcommand{\Tsig}{\mathrm{T}}
\newcommand{\Ssi}{\Ss_{u}}
\newcommand{\Sss}{\Ss_{g}}

\title[]{On asymptotic characterization of destabilizing switching signals for switched linear systems}

\author{Atreyee Kundu}
\address{Department of Electrical Engineering, Indian Institute of Science Bangalore, Bengaluru - 560012, India. }
\email{atreyeek@iisc.ac.in}

\date{\today}

\begin{document}



\begin{abstract}
    This paper deals with classes of (de)stabilizing switching signals for switched systems. Most of the available conditions for stability of switched systems are sufficient in nature, and consequently, their violation does not conclude instability of a switched system. The study of instability is, however, important for obvious reasons.  Our contributions are twofold: Firstly, we propose a class of switching signals under which a continuous-time switched linear system is unstable. Our characterization of instability depends solely on the asymptotic behaviour of frequency of switching, frequency of transition between subsystems, and fraction of activation of subsystems. Secondly, we show that our class of destabilizing switching signals is a strict subset of the class of switching signals that does not satisfy asymptotic characterization of stability recently proposed in the literature. This observation identifies a gap between asymptotic characterizations of stabilizing and destabilizing switching signals for switched linear systems. The main apparatus for our analysis is multiple Lyapunov-like functions.
\end{abstract}

%
%
  \maketitle
\section{Introduction}
\label{s:intro}
	A \emph{switched system} has two ingredients --- a family of systems and a switching signal. The \emph{switching signal} selects an \emph{active subsystem} at every instant of time, i.e., the system from the family that is currently being followed \cite[\S 1.1.2]{Liberzon}. Switched systems find wide applications in power systems and power electronics, automotive control, aircraft and air traffic control, network and congestion control, etc. \cite[p.\ 5]{Sun}. In this paper we will focus on continuous-time switched linear systems.
\subsection{Motivation}
\label{s:motive}
	It is well-known that qualitative properties of a switched system depends not only on the properties of its individual subsystems, but also on the properties of switching signals. In particular, divergent trajectories may be generated by switching appropriately among stable subsystems, while a suitably constrained switching signal may ensure stability of a switched system even if all subsystems are unstable (see e.g., \cite[p.\ 19]{Liberzon} for examples with two subsystems). Due to these interesting features, the problem of characterizing classes of switching signals that preserve stability of a switched system, has attracted a considerable research attention in the past few decades, see \cite{Antsaklis_survey, Heemels_survey, Liberzon_survey} for detailed surveys. We will restrict ourselves to purely time-dependent switching signals characterized by employing multiple Lyapunov-like functions \cite{Branicky}.
	
	Research on this topic can be broadly classified into two directions: stability characterization based on \emph{point-wise} properties of switching signals \cite{Morse'96,Hespanha'99, Liberzon'12, abc} and stability characterization based on \emph{asymptotic} properties of switching signals \cite{def,pqr}. In case of the former, stabilizing switching signals obey certain upper bounds on the number of switches and duration of activation of unstable subsystems on \emph{every} interval of time, while in case of the latter, stability is characterized based solely on the asymptotic properties of the switching signals. Recently in \cite{pqr} it was shown that if a switching signal satisfies any of the existing point-wise or asymptotic characterizations of stability, then it satisfies certain conditions on the asymptotic behaviour of switching frequency, frequency of transitions between subsystems, and fraction of activation of subsystems.
	
	The above body of results is derived by employing multiple Lyapunov-like functions, and is only sufficient in nature. Consequently, if a switching signal does not obey these stability conditions, we cannot conclude that the resulting switched system is unstable. This fact motivates the current work. We are interested in instability characterizations of switched systems using multiple Lyapunov-like functions.
	
	Instability is an important concept in stability theory. This concept is often useful in studying behaviour of a switched system under failures in system components, adversarial attacks, etc. It is known that if a family of asymptotically stable systems does not admit a common Lyapunov function, then the family admits at least one switching signal that is destabilizing. In \cite{Shorten'00} a sufficient condition for existence of a stabilizing switching signal was proposed using matrix pencils. Several classes of subsystems that admit destabilizing switching signals were identified. This set of results was employed to study the connection between existence of a destabilizing switching signal and non-existence of a common quadratic Lyapunov function in \cite{King'06}. In \cite{Wulff'11} a necessary but not sufficient condition for instability of a planar switched linear system was proposed by employing flow relations of the constituent subsystems and construction of invariant sets. Instability of stochastic switched systems under arbitrary switching was addressed in \cite{Zhang'15}. The authors proposed sufficient conditions for instability in a probabilistic sense.
	
	In this paper we focus on the problem of characterizing switching signals that are destabilizing.  Identification of such switching signals are useful for obvious reasons. For example, from the viewpoint of resilience of a system, it is often necessary to know which switching signals do not preserve good qualitative behaviour of the system so that adequate measures for preventing them may be taken. 
	
	We will study the following problem: given a family of continuous-time linear systems containing both asymptotically stable and unstable subsystems, and a set of admissible transitions, characterize a class of switching signals under which the resulting switched system is unstable. Since we allow the presence of unstable subsystems, constant switching signals on these subsystems are destabilizing, and constitute a trivial solution to our problem. However, our objective is to characterize a (possibly large) class of switching signals beyond the class of constant signals on unstable subsystems.

\subsection{Our contributions}
\label{s:contri}
	Multiple Lyapunov-like functions are a widely used tool for studying stability of switched systems \cite[Chapter 3]{Liberzon}. The underlying idea is that the maximum increase in these functions caused by activation of unstable subsystems and occurrence of switches is compensated by the minimum decrease caused by activation of asymptotically stable subsystems. We will utilize minimum increase and maximum decrease of these functions to characterize instability of a switched system. Our characterization of destabilizing switching signals involves asymptotic behaviour of the following properties of these signals: frequency of switching, frequency of transition between subsystems, and fraction of activation of subsystems. It does not involve nor imply conditions on a switching signal on \emph{every} interval of time.
	
	Earlier in \cite{pqr} asymptotic properties of switching signals were used to characterize stability of a switched system. Our class of destabilizing switching signals is a strict subset of the class of switching signals that does not obey the stability condition proposed in \cite[Theorem 5]{pqr}. (De)stabilizing properties of the class of switching signals that satisfies neither the stability condition of \cite{pqr} nor the instability condition proposed in this paper, remains undetermined.
	
	To summarize, our contributions in this paper are the following:
	\begin{itemize}[label = \(\circ\),leftmargin = *]
		\item Given a family of systems containing both asymptotically stable and unstable subsystems, and a set of admissible transitions, we characterize a class of switching signals under which the resulting switched system is unstable. Our characterization of instability is based solely on the asymptotic properties of a switching signal.
		\item We study relations between asymptotic characterization of stabilizing and destabilizing switching signals for switched systems.
	\end{itemize}
	To the best of our knowledge, this is the first instance in the literature when multiple Lyapunov-like functions are employed to characterize instability of a switched system, and the gap between asymptotic characterization of stabilizing and destabilizing switching signals, is addressed.
	
	The remainder of this paper is organized as follows: in \S\ref{s:prelims} we formulate the problem under consideration, and catalog required definitions and notations. We also recall the stability condition from \cite{pqr} in this section. Our main results appear in \S\ref{s:mainres}. We elaborate on our results with discussions and numerical examples. We conclude in \S\ref{s:concln} with a brief discussion of future research directions. Proofs for the auxiliary and main results are presented in a consolidated fashion in \S\ref{s:all_proofs1} and \S\ref{s:all_proofs2}, respectively.
	
	{\bf Notation}. \(\R\) is the set of real numbers and \(\norm{\cdot}\) is the standard \(2\)-norm. We let \(\mathrm{1}(A)\) denote the indicator function of a set \(A\) and \(^\top\) denote the transpose operation. For a square matrix \(M\), we denote by \(\lambda_{\min}(M)\) and \(\lambda_{\max}(M)\) the minimum and maximum eigenvalues of \(M\), respectively. \(A^{c}\) denotes the complement set of a set \(A\).
\section{Preliminaries}
\label{s:prelims}
\subsection{Switched linear systems}	
\label{s:sw_lin_sys}
	We consider a family of continuous-time linear systems
	\begin{align}
	\label{e:family}
		\dot{x}(t) = A_{p}x(t),\:\:x(0) = x_{0},\:\:p\in\P,\:\:t\in[0,+\infty[,
	\end{align}
	where \(x(t)\in\R^{d}\) is the vector of states at time \(t\), and \(\P = \{1,2,\ldots,N\}\) is an index set. We assume that for each \(p\in\P\) the matrix \(A_{p}\in\R^{d\times d}\) has full rank; consequently, \(0\in\R^{d}\) is the unique equilibrium point for each system in \eqref{e:family}. Let \(\sigma:[0,+\infty[\to\P\) be a switching signal; it is a piecewise constant function that specifies, at each time \(t\), the subsystem \(A_{\sigma(t)}\), that is active at \(t\). By convention, \(\sigma\) is assumed to be continuous from right and having limits from the left everywhere. A switched system generated by the family of systems \eqref{e:family} and a switching signal \(\sigma\) is given by
	\begin{align}
	\label{e:swsys}
		\dot{x}(t) = A_{\sigma(t)}x(t),\:\:x(0) = x_{0},\:\:t\in[0,+\infty[.
	\end{align}
	
	Let \(\P_{S}\) and \(\P_{U}\) denote the sets of indices of asymptotically stable and unstable subsystems, respectively, \(\P = \P_{S}\sqcup\P_{U}\), and \(E(\P)\) denote the set of all ordered pairs \((p,q)\) such that a transition from subsystem \(p\) to subsystem \(q\) is admissible, \(p,q\in\P\).\footnote{By asymptotically stable subsystems, we mean that the matrices \(A_{p}\)'s are Hurwitz, and for the unstable subsystems, \(A_{p}\)'s are not Hurwitz.} We let \(0 =: \tau_{0}<\tau_{1}<\tau_{2}<\cdots\) be the switching instants; these are the points in time where \(\sigma\) jumps. We call a switching signal \(\sigma\) admissible if it satisfies \((\sigma(\tau_{i}),\sigma(\tau_{i+1}))\in E(\P)\), \(i=0,1,\ldots\). Let \(\mathcal{S}\) denote the set of all admissible switching signals. For \(t > 0\), let \(\Ntsigma(t)\) denote the number of switches on \(]0,t]\). The solution \((x(t))_{t\geq 0}\) to the switched system \eqref{e:swsys} corresponding to an admissible switching signal \(\sigma\in\mathcal{S}\) is the map \(x:[0,+\infty[\to\R^{d}\) defined by
	\begin{align*}
		x(t) = e^{A_{\sigma(\tau_{\Ntsigma(t)})}(t-\tau_{\Ntsigma(t)})}\cdots e^{A_{\sigma(\tau_{0})}(\tau_{1}-\tau_{0})}x_{0},
	\end{align*}
	where the dependence of \(x\) on \(\sigma\) is suppressed for notational simplicity.
	\begin{definition}
	\label{d:gas}
	\rm{
		The switched system \eqref{e:swsys} is \emph{globally asymptotically stable} (GAS) for a given switching signal \(\sigma\) if \eqref{e:swsys} is
		\begin{itemize}[label = \(\circ\), leftmargin = *]
			\item Lyapunov stable, and
			\item globally asymptotically convergent, i.e., for all \(x(0)\), \(\norm{x(t)}\to 0\) as \(t\to+\infty\).
		\end{itemize}
	}
	\end{definition}
	\begin{definition}
	\label{d:uns}
	\rm{
		The switched system \eqref{e:swsys} is \emph{unstable} for a given switching signal \(\sigma\) if for all \(x(0)\neq 0\), \(\norm{x(t)}\to+\infty\) as \(t\to+\infty\).
	}
	\end{definition}
	
\subsection{Stabilizing switching signals}
\label{ss:stab_sw}
	Given a family of systems \eqref{e:family} and a set of admissible transitions \(E(\P)\), characterizing a class of switching signals under which the resulting switched system \eqref{e:swsys} is GAS, is studied widely in the literature \cite[Chapter 3]{Liberzon}. A useful tool for this study is multiple Lyapunov-like functions \cite{Branicky}. Below we recall these functions and their properties.
	\begin{fact}{\cite[Fact 2.3]{arxiv'13}}
		\label{fact:key1}
		\rm{
			For each \(p\in\P\) there exists a pair \((P_{p},\hat{\lambda}_{p})\), where \(P_{p}\in\R^{d\times d}\) is a symmetric and positive definite matrix, and
			\begin{itemize}[label = \(\circ\), leftmargin = *]
				\item if \(A_{p}\) is asymptotically stable, then \(\hat{\lambda}_{p} > 0\),
				\item if \(A_{p}\) is unstable, then \(\hat{\lambda}_{p} \leq 0\),
			\end{itemize}
			such that, with
			\begin{align}
			\label{e:Lyap_func-defn}
				\R^{d}\ni\xi\mapsto V_{p}(\xi) := <P_{p}\xi,\xi>\in[0,+\infty[,
			\end{align}
			we have for all \(\gamma_{p}(0)\in\R^{d}\), \(t\in[0,+\infty[\),
			\begin{align}
			\label{e:key_ineq1}
				V_{p}(\gamma_{p}(t))\leq \exp(-\hat{\lambda}_{p}t)V_{p}(\gamma_{p}(0)),
			\end{align}
			and \(\gamma_{p}(\cdot)\) solves the \(p\)-th system dynamics in \eqref{e:family}, \(p\in\P\).
		}
		\end{fact}
		The functions \(V_{p}\), \(p\in\P\) are called Lyapunov-like functions. The scalar \(\hat{\lambda}_{p}\) gives a quantitative measure of (in)stability of the \(p\)-th subsystem, \(p\in\P\). Indeed, the inequality \eqref{e:key_ineq1} captures the minimum rate of decay and maximum rate of growth of \(V_{p}\) for the asymptotically stable and unstable subsystems, respectively.
		
		By definition of the Lyapunov-like functions \(V_{p}\), \(p\in\P\) in \eqref{e:Lyap_func-defn}, they are linearly comparable. The following fact gives a measure of the maximum increase of the Lyapunov-like functions caused by a transition from subsystem \(p\) to subsystem \(q\).
		\begin{fact}{\cite[\S2.1]{arxiv'13}}
		\label{fact:key2}
		\rm{
			For each \((p,q)\in E(\P)\), the respective Lyapunov-like functions are related as follows: there exists \(\hat{\mu}_{pq} \geq 1\) such that
			\begin{align}
			\label{e:key_ineq2}
			 	V_{q}(\xi) \leq 	\hat{\mu}_{pq}V_{p}(\xi)\:\:\text{for all}\:\xi\in\R^{d}.
			\end{align}
		}
		\end{fact}
		In \cite[Proposition 2.4]{arxiv'13} a \emph{tight} estimate of \(\hat{\mu}_{pq}\), \((p,q)\in E(\P)\) was given as \(\lambda_{\max}(P_{q}P_{p}^{-1})\).
		
		Let us now introduce some notations related to a switching signal.
		\begin{itemize}[label = \(\circ\), leftmargin = *]
			\item Fix \(t > 0\). Let
				\begin{align}
				\label{e:sw_freq}
					\nu(t) := \frac{\Ntsigma(t)}{t}
				\end{align}
			be the \emph{frequency of switching} at \(t\).
			\item We let \(\Ntsigma_{pq}(t)\) denote the number of times a switch from subsystem \(p\) to subsystem \(q\) has occurred before (and including) time \(t\). It follows that
			\(\displaystyle{\Ntsigma(t) = \sum_{(p,q)\in E(\P)}\Ntsigma_{pq}(t)}\). Let
			\begin{align}
			\label{e:trans_freq}
				\rho_{pq}(t) := \frac{\Ntsigma_{pq}(t)}{\Ntsigma(t)}
			\end{align}
			be the \emph{transition frequency from subsystem \(p\) to subsystem \(q\)} on \(]0,t]\), \((p,q)\in E(\P)\).
			\item We let \(\Tsig_{p}(t)\) denote the total duration of activation of subsystem \(p\) on \(]0,t]\), \(p\in\P\). Let
			\begin{align}
			\label{e:frac_activ}
				\eta_{p}(t) := \frac{\Tsig_{p}(t)}{t}
			\end{align}
			denote the \emph{fraction of activation of subsystem \(p\)} on the interval \(]0,t]\), \(p\in\P\).
		\end{itemize}

		A class of switching signals that ensures GAS of the switched linear system \eqref{e:swsys} was proposed in \cite{pqr}.
		
		\begin{theorem}{\cite[Theorem 5]{pqr}}
		\label{t:stab_thm}
		\rm{
		Consider a family of systems \eqref{e:family}. The switched system \eqref{e:swsys} is globally asymptotically stable (GAS) for every switching signal \(\sigma\in\Ss\) that satisfies
		\begin{align}
		\label{e:stab_condn}
			\varlimsup_{t\to+\infty}\Biggl(\nu(t)\sum_{(p,q)\in E(\P)}(\ln\hat{\mu}_{pq})\rho_{pq}(t) &- \sum_{p\in\P_{S}}\abs{\hat{\lambda}_{p}}\eta_{p}(t)
			+ \sum_{p\in\P_{U}}\abs{\hat{\lambda}_{p}}\eta_{p}(t) \Biggr) < 0,
		\end{align}
		where \(\hat{\lambda}_{p}\), \(p\in\P\) and \(\hat{\mu}_{pq}\), \((p,q)\in E(\P)\) are as in Facts \ref{fact:key1} and \ref{fact:key2}, respectively, and \(\nu(t)\), \(\rho_{pq}(t)\), \((p,q)\in E(\P)\) and \(\eta_{p}(t)\), \(p\in\P\) are as defined in \eqref{e:sw_freq}, \eqref{e:trans_freq} and \eqref{e:frac_activ}, respectively.
		}
	\end{theorem}
	
	Theorem \ref{t:stab_thm} contains a class of switching signals under which the switched system \eqref{e:swsys} is GAS. The characterization of stability relies solely on the asymptotic properties of a switching signal \(\sigma\). Fix \(t > 0\). The term \(\displaystyle{\nu(t)\sum_{(p,q)\in E(\P)}(\ln\hat{\mu}_{pq})\rho_{pq}(t)}\) involves the frequency of switching of \(\sigma\) at \(t\) and the frequency of transitions from subsystem \(p\) to subsystem \(q\) on \(]0,t]\) weighted by a measure of increase or decrease in the corresponding Lyapunov-like functions \(V_{p}\) and \(V_{q}\), while the sum \(\displaystyle{- \sum_{p\in\P_{S}}\abs{\hat{\lambda}_{p}}\eta_{p}(t)+\sum_{p\in\P_{U}}\abs{\hat{\lambda}_{p}}\eta_{p}(t)}\) contains fraction of activation of asymptotically stable and unstable subsystems on \(]0,t]\) weighted by a measure of their (in)stability. If a switching signal \(\sigma\) satisfies that the sum
	\[
		\nu(t)\sum_{(p,q)\in E(\P)}(\ln\hat{\mu}_{pq})\rho_{pq}(t) - \sum_{p\in\P_{S}}\abs{\hat{\lambda}_{p}}\eta_{p}(t)
			+\sum_{p\in\P_{U}}\abs{\hat{\lambda}_{p}}\eta_{p}(t)
	\]
	is strictly less than zero in the asymptote, then it is stabilizing. The properties of \(\sigma\) on \emph{every} interval of time is not considered. In the sequel we let \(\Sss\subset\Ss\) denote the set of all switching signals \(\sigma\) that satisfies condition \eqref{e:stab_condn}.
	
	\begin{remark}
	\label{rem:why_prq}
	\rm{
		From the set of sufficient conditions for stability of switched systems available in the literature, we selected \cite[Theorem 5]{pqr} as reference due to its generality. It was shown in \cite[Theorem 6]{pqr} that if a switching signal satisfies (average) dwell time condition \cite{Hespanha'99, chatterjee07, Liberzon'12} or asymptotic frequency and dwell time condition \cite{abc}, then it satisfies the stability condition \eqref{e:stab_condn}. 
	}
	\end{remark}
	
	\begin{remark}
	\label{rem:non-unique_scalars}
	\rm{
		Notice that given the matrices \(A_{p}\), \(p\in\P\), the choice of the Lyapunov-like functions \eqref{e:Lyap_func-defn} and consequently, the scalars \(\hat{\lambda}_{p}\), \(p\in\P\) and \(\hat{\mu}_{pq}\), \((p,q)\in E(\P)\) is not unique. In the sequel, for a fixed switching signal \(\sigma\in\Ss\), we will verify conditions that involve the scalars \(\hat{\lambda}_{p}\), \(p\in\P\) and \(\hat{\mu}_{pq}\), \((p,q)\in E(\P)\) with respect to ``given'' matrices \(P_{p}\), \(p\in\P\).
	}
	\end{remark}
	
	\begin{cor}
	\label{cor:instab_cor}
	\rm{
		Consider a family of systems \eqref{e:family}. If a switching signal \(\sigma\) ensures instability of the switched system \eqref{e:swsys}, then it satisfies
		\begin{align}
		\label{e:instab_condn1}
			\varlimsup_{t\to+\infty}\Biggl(\nu(t)\sum_{(p,q)\in E(\P)}(\ln\hat{\mu}_{pq})\rho_{pq}(t) &- \sum_{p\in\P_{S}}\abs{\hat{\lambda}_{p}}\eta_{p}(t)
			+ \sum_{p\in\P_{U}}\abs{\hat{\lambda}_{p}}\eta_{p}(t) \Biggr) \geq 0,
		\end{align}
		where \(\hat{\lambda}_{p}\), \(p\in\P\) and \(\hat{\mu}_{pq}\), \((p,q)\in E(\P)\) are as in Facts \ref{fact:key1} and \ref{fact:key2}, respectively, and \(\nu(t)\), \(\rho_{pq}(t)\), \((p,q)\in E(\P)\) and \(\eta_{p}(t)\), \(p\in\P\) are as defined in \eqref{e:sw_freq}, \eqref{e:trans_freq} and \eqref{e:frac_activ}, respectively.
	}	
	\end{cor}
	Corollary \ref{cor:instab_cor} is a contrapositive statement of Theorem \ref{t:stab_thm}. Notice that condition \eqref{e:instab_condn1} is a necessary condition for instability of \eqref{e:swsys}, but is not a sufficient condition.
	
	\begin{example}
	\label{ex:counter_ex1}
	\rm{
	Consider a family of systems \eqref{e:family} with \(\P = \{1,2\}\), where \(A_{1} = \pmat{-0.2 & -0.4\\3 & -0.2}\) and \(A_{2} = \pmat{-0.2 & -3\\0.4 & -0.2}\). Clearly, \(\P_{S} = \{1,2\}\) and \(\P_{U} = \emptyset\). Let \(E(\P) = \{(1,2),(2,1)\}\).
	
	We fix the Lyapunov-like functions \(V_{p}(\xi) = <P_{p}\xi,\xi>\), \(p\in\P\) with \(P_{1} = \pmat{10.3629 & 0.5242\\0.5242 & 1.4516}\) and \(P_{2} = \pmat{1.4516 & -0.5242\\-0.5242 & 10.3629}\). The scalars \(\hat{\lambda}_{p}\), \(p\in\P\) and \(\hat{\mu}_{pq}\), \((p,q)\in E(\P)\) are obtained from \(P_{p}\), \(p\in\P\) as follows:{\footnote{By selecting \(Q_{p} = I_{2\times 2}\), \(p\in\P\) and solving the Lyapunov equation \(A_{p}^\top P_{p} + P_{p}A_{p} = -Q_{p}\), we obtain the said choice of \(P_{p}\), \(p\in\P\). The scalars \(\hat{\lambda}_{p}\), \(p\in\P\) are computed as \(\displaystyle{\hat{\lambda}_{p} = \frac{\lambda_{\min}(Q_{p})}{\lambda_{\max}(P_{p})}}\). Following \cite[Proposition 2.4]{arxiv'13} we compute \(\hat{\mu}_{pq}\), \((p,q)\in E(\P)\) as \(\lambda_{\max}(P_{q}P_{p}^{-1})\).}}
	\[
		\hat{\lambda}_{1} = \hat{\lambda}_{2} = 0.0962,\quad\hat{\mu}_{12} = \hat{\mu}_{21} = 7.3149.
	\]
	 Now, consider a switching signal \(\sigma\) that satisfies:
	 \begin{itemize}[label = \(\circ\),leftmargin=*]
	 	\item \(\displaystyle{\Ntsigma(t) = \lfloor{\frac{t}{10}}\rfloor}\),
		\item \(\displaystyle{\Ntsigma_{pq}(t) = \frac{\Ntsigma(t)}{2}}\), \((p,q)\in E(\P)\), and
		\item \(\T_{p}(t) = \frac{t}{2}\), \(p\in\P\).
	\end{itemize}
	Clearly,
	 \begin{itemize}[label = \(\circ\),leftmargin=*]
		\item \(\displaystyle{\nu(t) = \frac{1}{10}}\),
		\item \(\displaystyle{\rho_{pq}(t) = \frac{1}{2}}\), \((p,q)\in E(\P)\), and
		\item \(\displaystyle{\eta_{p}(t) = \frac{1}{2}}\), \(p\in\P\).
	\end{itemize}
		We have
	\begin{align*}
		\varlimsup_{t\to+\infty}\Biggl(&\nu(t)\sum_{(p,q)\in E(\P)}(\ln\hat{\mu}_{pq})\rho_{pq}(t) - \sum_{p\in\P_{S}}\abs{\hat{\lambda}_{p}}\eta_{p}(t)
		+ \sum_{p\in\P_{U}}\abs{\hat{\lambda}_{p}}\eta_{p}(t) \Biggr)= 0.1028 > 0.
	\end{align*}
	However, we observe that \(\sigma\) is stabilizing for \eqref{e:swsys}. See Figures \ref{fig:sw_plot2} and \ref{fig:x_plot2} for illustrations of \((\sigma(t))_{t\geq 0}\) and \((\norm{x(t)})_{t\geq 0}\) till \(t = 120\) units of time, respectively. The initial condition \(x(0)\) is chosen as \(\pmat{-1.0883\\2.9263}\).
	\begin{figure}[htbp]
			\includegraphics[scale = 0.45]{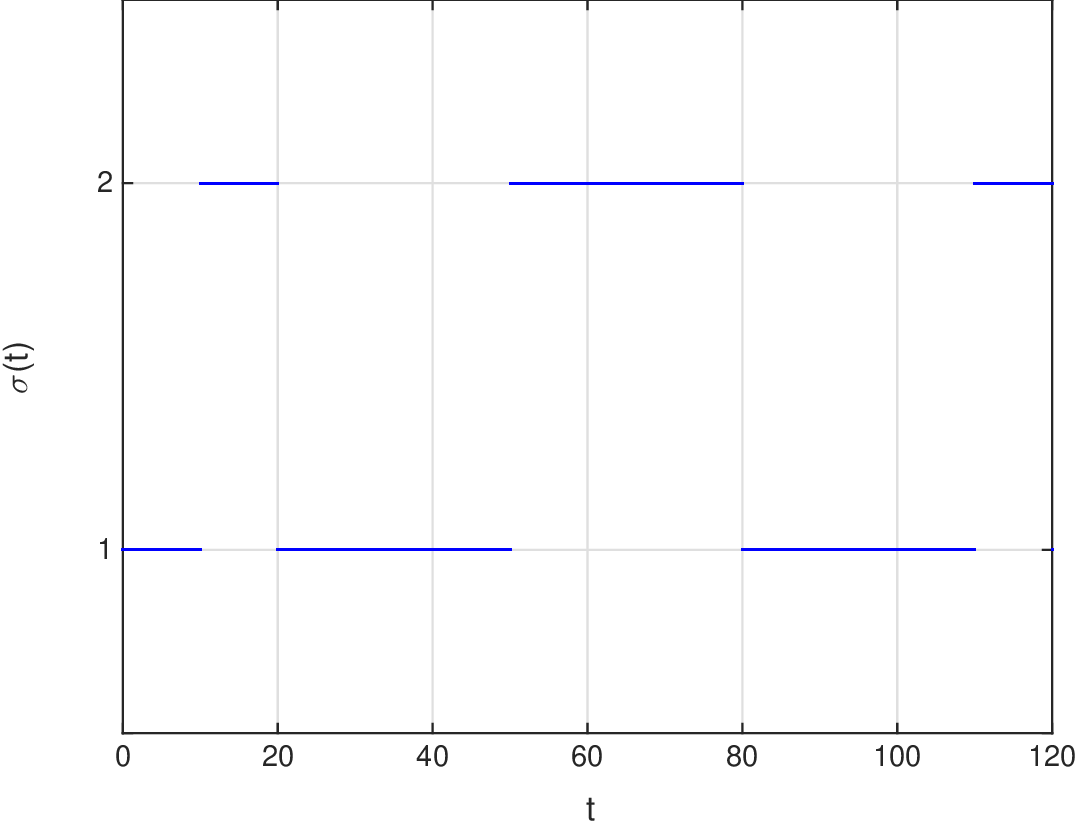}
			\caption{Plot of \((\sigma(t))_{t\geq 0}\) for Example \ref{ex:counter_ex1}} \label{fig:sw_plot2}
		\end{figure}
		\begin{figure}[htbp]
			\includegraphics[scale = 0.45]{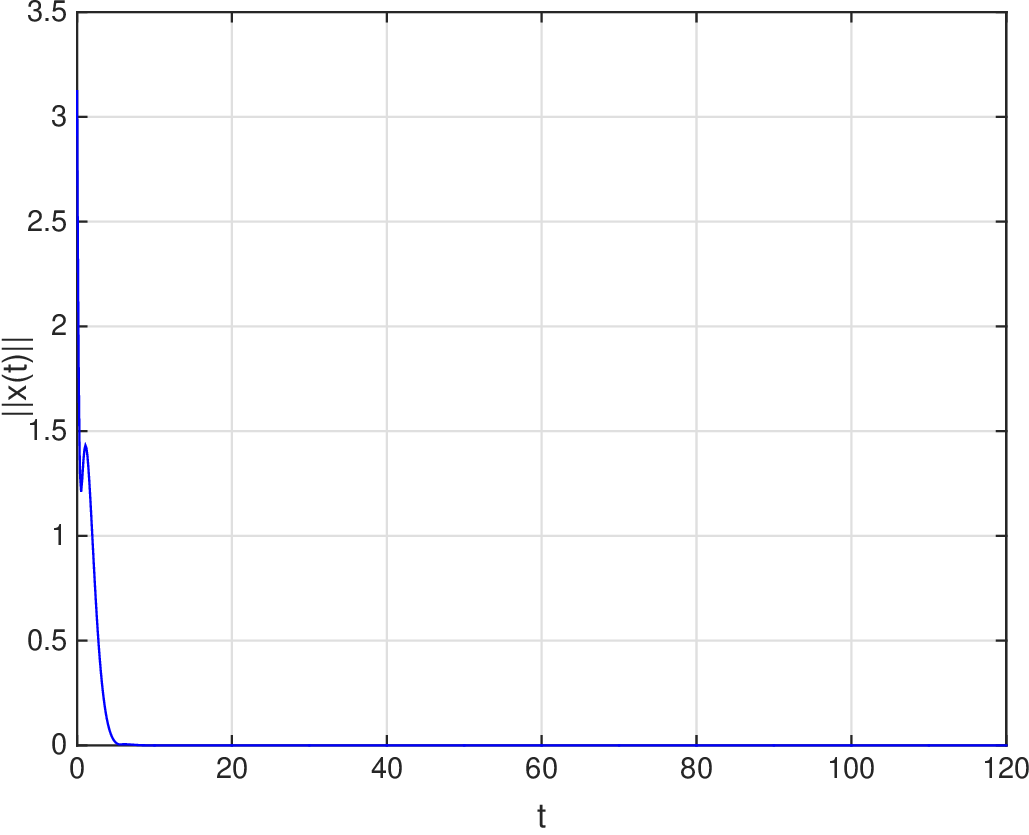}
			\caption{Plot of \((\norm{x(t))})_{t\geq 0}\) for Example \ref{ex:counter_ex1}} \label{fig:x_plot2}
		\end{figure}
	}
	\end{example}
	\begin{remark}
	\label{rem:non-unique-scalars}
	\rm{
		As highlighted in Remark \ref{rem:non-unique_scalars}, non-satisfaction of condition \eqref{e:instab_condn1} in the above example is for a fixed choice of the matrices \(P_{p}\), \(p\in\P\). To wit, we do not claim that the family of systems \eqref{e:family} in Example \ref{ex:counter_ex1} does not admit multiple Lyapunov-like functions such that condition \eqref{e:instab_condn1} holds. Ideally, given a family of systems \eqref{e:family} and a switching signal \(\sigma\), one would like to determine if there exists a choice of \(P_{p}\), \(p\in\P\) such that with the corresponding scalars \(\hat{\lambda}_{p}\), \(p\in\P\) and \(\hat{\mu}_{pq}\), \((p,q)\in E(\P)\), condition \eqref{e:instab_condn1} holds. To the best of our knowledge, this selection problem is computationally difficult due to the non-convexity associated with the computation of \(\hat{\lambda}_{p}\), \(p\in\P\) and \(\hat{\mu}_{pq}\), \((p,q)\in E(\P)\) beyond certain restrictive conditions on the matrices \(P_{p}\), \(p\in\P\). Checking properties of switching signals with respect to given multiple Lyapunov-like functions is standard in the literature, see e.g., \cite{mitra_ADT}.
	}
	\end{remark}
	\begin{remark}
	\label{rem:why_suff}
	\rm{
		Stability analysis of switched systems employing multiple Lyapunov-like functions involves compensating the maximum increase of these functions, caused by activation of unstable subsystems and occurrence of switches, by the minimum decrease in these functions, caused by activation of asymptotically stable subsystems. Therefore, the scalars \(\hat{\lambda}_{p}\), \(p\in\P\) and \(\hat{\mu}_{pq}\), \((p,q)\in E(\P)\) are useful in this study. However, the inequalities \eqref{e:key_ineq1} and \eqref{e:key_ineq2} are not sufficient to study instability of the switched system \eqref{e:swsys}. Indeed, divergence of \(\norm{x(t)}\) for all \(x(0)\) under the ``worst'' case switching is not sufficient for instability.
	}
	\end{remark}
	
\subsection{The problem}
\label{ss:prob_defn}
	The problem of our interest is the following:
	\begin{problem}
	\label{prob:mainprob}
	\rm{
		Given a family of systems \eqref{e:family} containing both asymptotically stable and unstable subsystems, and a set of admissible transitions \(E(\P)\), characterize a class of switching signals \(\mathcal{S}_{u}\subset\mathcal{S}\) under which the switched linear system \eqref{e:swsys} is unstable.
	}
	\end{problem}
	
	\begin{remark}
	\label{rem: non_trivial}
	\rm{
		Since \(\P_{U}\neq\emptyset\), a switching signal \(\sigma\) satisfying \(\sigma(t) = p\) for all \(t\) with a fixed \(p\in\P_{U}\) is destabilizing. However, we are seeking for a (possibly large) class of destabilizing switching signals not restricted to the class of constant switching signals.
	}
	\end{remark}

			Intuitively, for instability of the switched system \eqref{e:swsys} it is sufficient to capture the minimum increase and maximum decrease of the Lyapunov-like functions \(V_{p}\), \(p\in\P\), and ensure divergence of \(\norm{x(t)}\) for all \(x(0)\) under the ``best'' case switching. We will follow this route.
		

\subsection{Auxiliary results}
\label{ss:aux_res}
	\begin{fact}
		\label{fact:key3}
		\rm{
			For each \(p\in\P\) there exists a pair \((P_{p},\check{\lambda}_{p})\), where \(P_{p}\in\R^{d\times d}\) is a symmetric and positive definite matrix, and
			\begin{itemize}[label = \(\circ\), leftmargin = *]
				\item if \(A_{p}\) is asymptotically stable, then \(\check{\lambda}_{p} > 0\),
				\item if \(A_{p}\) is unstable, then \(\check{\lambda}_{p} \leq 0\),
			\end{itemize}
			such that, with the function \(V_{p}:\R^{d}\to[0,+\infty[\) defined as in \eqref{e:Lyap_func-defn},
			we have for all \(\gamma_{p}(0)\in\R^{d}\), \(t\in[0,+\infty[\),
			\begin{align}
			\label{e:key_ineq3}
				V_{p}(\gamma_{p}(t))\geq \exp(-\check{\lambda}_{p}t)V_{p}(\gamma_{p}(0)),
			\end{align}
			and \(\gamma_{p}(\cdot)\) solves the \(p\)-th system dynamics in \eqref{e:family}, \(p\in\P\).
		}
		\end{fact}
		
	Fact \ref{fact:key3} uses Lyapunov-like functions as defined in \eqref{e:Lyap_func-defn}, and captures the maximum rate of decay and minimum rate of growth of \(V_{p}\) for \(p\in\P_{S}\) and \(p\in\P_{U}\), respectively. The scalar \(\check{\lambda}_{p}\), \(p\in\P\) gives a quantitative measure of (in)stability of the \(p\)-th subsystem in the above sense. We provide a short proof of Fact \ref{fact:key3} in \S\ref{s:all_proofs1}.
	
	We also require a measure of the minimum increase between the Lyapunov-like functions \(V_{p}\) and \(V_{q}\) caused by a transition from subsystem \(p\) to subsystem \(q\).
		\begin{fact}
		\label{fact:key4}
		\rm{
			For each \((p,q)\in E(\P)\), the respective Lyapunov-like functions are related as follows: there exists \(\check{\mu}_{pq} > 0\) such that
			\begin{align}
			\label{e:key_ineq4}
			 	V_{q}(\xi) \geq 	\check{\mu}_{pq}V_{p}(\xi)\:\:\text{for all}\:\xi\in\R^{d}.
			\end{align}
		}
		\end{fact}
		Along the lines of Remark \ref{rem:non-unique_scalars}, the choice of the scalars \(\check{\lambda}_{p}\), \(p\in\P\) and \(\check{\mu}_{pq}\), \((p,q)\in E(\P)\) is not unique.  In the spirit of \cite[Proposition 2.4]{arxiv'13} we propose a \emph{tight} estimate of \(\check{\mu}_{pq}\), \((p,q)\in E(\P)\) as follows:
		\begin{prop}
		\label{prop:mu_est}
		\rm{
			Let the Lyapunov-like functions be defined as in \eqref{e:Lyap_func-defn} with each \(P_{p}\) symmetric and positive definite, \(p\in\P\). Then \(\check{\mu}_{pq}\) in \eqref{e:key_ineq2} can be computed as
			\begin{align}
			\label{e:mu_est}
				\check{\mu}_{pq} = \lambda_{\min}(P_{q}P_{p}^{-1}),\:\:(p,q)\in E(\P).
			\end{align}
		}
		\end{prop}
		
		A proof of Proposition \ref{prop:mu_est} is given in \S\ref{s:all_proofs1}. The scalars \(\check{\lambda}_{p}\), \(p\in\P\) and \(\check{\mu}_{pq}\), \((p,q)\in E(\P)\) are relevant to our analysis for instability.
\section{Main results}
\label{s:mainres}	
	We are now in a position to present a class of destabilizing switching signals for switched linear systems.	
\subsection{Destabilizing switching signals}
\label{ss:destab_sw}
	\begin{theorem}
	\label{t:instab_thm}
	\rm{
		Consider a family of systems \eqref{e:family}. The switched system \eqref{e:swsys} is unstable for every switching signal \(\sigma\in\Ss\) that satisfies
		\begin{align}
		\label{e:instab_condn}
			\varliminf_{t\to+\infty}\Biggl(\nu(t)\sum_{(p,q)\in E(\P)}(\ln\check{\mu}_{pq})\rho_{pq}(t) &- \sum_{p\in\P_{S}}\abs{\check{\lambda}_{p}}\eta_{p}(t)
			+\sum_{p\in\P_{U}}\abs{\check{\lambda}_{p}}\eta_{p}(t)\Biggr) > 0,
		\end{align}
		where \(\check{\lambda}_{p}\), \(p\in\P\) and \(\check{\mu}_{pq}\), \((p,q)\in E(\P)\) are as in Facts \ref{fact:key3} and \ref{fact:key4}, respectively, and \(\nu(t)\), \(\rho_{pq}(t)\), \((p,q)\in E(\P)\) and \(\eta_{p}(t)\), \(p\in\P\) are as defined in \eqref{e:sw_freq}, \eqref{e:trans_freq} and \eqref{e:frac_activ}, respectively.
	}
	\end{theorem}
	
	Given a family of systems \eqref{e:family} containing both asymptotically stable and unstable subsystems, and a set of admissible transitions \(E(\P)\), in Theorem \ref{t:instab_thm} we characterize a class of switching signals under which the switched system \eqref{e:swsys} is unstable. We will call the set of all \(\sigma\) that satisfy condition \eqref{e:instab_condn} as  \(\Ssi\). As in the case of Theorem \ref{t:stab_thm} for stability of \eqref{e:swsys}, our characterization of instability of \eqref{e:swsys} relies on the asymptotic behaviour of frequency of switching, frequency of transition between subsystems, and fraction of activation of subsystems.
		
		In \eqref{e:instab_condn}, the fractions of activation of subsystems, \(\eta_{p}(t)\), \(p\in\P\), are weighted by the maximum rate of decay and minimum rate of growth of the corresponding Lyapunov-like function \(V_{p}\), \(p\in\P\) for asymptotically stable and unstable subsystems, respectively. The weighing factors for frequency of transition between subsystems, \(\rho_{pq}(t)\), \((p,q)\in E(\P)\) are the scalars \(\check{\mu}_{pq}\), \((p,q)\in E(\P)\) that give a measure of the minimum ``jump'' in the corresponding Lyapunov-like functions \(V_{p}\) and \(V_{q}\) caused by a transition from subsystem \(p\) to subsystem \(q\). This is in contrast to the minimum rate of decay and maximum rate of growth of \(V_{p}\), \(p\in\P\) employed in condition \eqref{e:stab_condn}.
	\begin{remark}
	\label{rem:condn_diff}
	\rm{
		On the one hand, estimates of minimum duration of activation of stable subsystems and maximum duration of activation of unstable subsystems are determined to guarantee convergence of \(\norm{x(t)}\) under worst case switching \cite{Hespanha'99,Liberzon'12,def}. On the other hand, we determine estimates of maximum duration of activation of stable subsystems and minimum duration of activation of unstable subsystems to guarantee divergencee of \(\norm{x(t)}\) under best case switching. Indeed, condition \eqref{e:stab_condn} involves a strict upper bound on the asymptotic behaviour of\\\(\displaystyle{\nu(t)\sum_{(p,q)\in E(\P)}(\ln\hat{\mu}_{pq})\rho_{pq}(t) - \sum_{p\in\P_{S}}\abs{\hat{\lambda}_{p}}\eta_{p}(t)
			+\sum_{p\in\P_{U}}\abs{\hat{\lambda}_{p}}\eta_{p}(t)}\),\\ while condition \eqref{e:instab_condn} involves a strict lower bound on the asymptotic behaviour of
\[
    \displaystyle{\nu(t)\sum_{(p,q)\in E(\P)}(\ln\check{\mu}_{pq})\rho_{pq}(t) - \sum_{p\in\P_{S}}\abs{\check{\lambda}_{p}}\eta_{p}(t)+}\displaystyle{\sum_{p\in\P_{U}}\abs{\check{\lambda}_{p}}\eta_{p}(t)}.
\]
	}
	\end{remark}
%
	\begin{remark}
	\label{rem:non_trivial_class}
	\rm{
		Let \(\P_{U} = \emptyset\). It is known that if all subsystems are linear and asymptotically stable, then a sufficiently ``slow'' switching signal between these subsystems preserves stability of the resulting switched system. On the one hand, the known estimates of ``slowness'' \cite{Morse'96,Hespanha'99} are sufficient in the sense that switching at a faster rate does not necessarily guarantee instability of the resulting switched system. On the other hand, condition \eqref{e:instab_condn} provides a measure of how ``fast'' one needs to switch between asymptotically stable subsystems such that the resulting switched system loses stability.
	}
	\end{remark}
	\begin{example}
	\label{ex:exmpl1}
	\rm{
		Consider a family of systems \eqref{e:family} with \(\P = \{1,2,3,4\}\), where
		\begin{align*}
			A_{1} &= \pmat{-0.1857 & -0.7565\\-0.0707 & -0.6500},\quad A_{2} &= \pmat{-0.3509 & -0.2683\\-0.3523 & -0.5491}\quad
			A_{3} &= \pmat{0.1734 & -0.6091\\0.8314 & -0.1966},\quad A_{4} &= \pmat{0.6294 & 0.8116\\-0.7460 & 0.8268}.
		\end{align*}
		Clearly, \(\P_{S} = \{1,2,3\}\) and \(\P_{U} = \{4\}\). Let \(E(\P) = \{(1,2),(1,3),(2,1),\)\((2,4),(3,1),(3,4),(4,2),(4,3)\}\).\\
		
		{\it A. Multiple Lyapunov-like functions}: We choose \(V_{p}(\xi) = \xi^\top P_{p}\xi\), \(p\in\P\). The pairs \((P_{p},\check{\lambda}_{p})\), \(p\in\P\) are computed in the following manner:
		\begin{itemize}[label = \(\circ\), leftmargin = *]
			\item For \(p\in\P_{S}\), we set \(Q_{p} = I_{2\times 2}\), solve the Lyapunov equation \(A_{p}^\top P_{p} + P_{p}A_{p} = -Q_{p}\) for \(P_{p}\), and compute \(\check{\lambda}_{p} = \frac{\lambda_{\max}(Q_{p})}{\lambda_{\min}(P_{p})}\).
			\item For \(p\in\P_{U}\), we set \(\varepsilon_{p} = Re(\lambda_{\max}(A_{p})) + 0.0001\) and \(Q_{p} = I_{2\times 2}\), solve the Lyapunov equation \((A_{p}-\varepsilon_{p}I)^\top P_{p} + P_{p}(A_{p}-\varepsilon_{p}I) = -Q_{p}\) for \(P_{p}\), and compute \(\check{\lambda}_{p} = 2\varepsilon_{p} - \frac{\lambda_{\max}(Q_{p})}{\lambda_{\min}(P_{p})}\).
		\end{itemize}
		We obtain
		\begin{align*}
			(P_{1},\check{\lambda}_{1}) &= \Biggl(\pmat{4.4041 & -4.4942\\-4.4942 & 5.9995}, 1.5691\Biggr),&\quad\quad
			(P_{2},\check{\lambda}_{2}) &= \Biggl(\pmat{2.9643 & -1.5333\\-1.5333 & 1.6598}, 1.5486\Biggr),\\
			(P_{3},\check{\lambda}_{3}) &= \Biggl(\pmat{54.7476 & -12.0193\\-12.0193 & 39.7715}, 0.0302\Biggr),&\quad\quad
			(P_{4},\check{\lambda}_{4}) &= \Biggl(\pmat{4.8542 & -0.6419\\-0.6419 & 5.2809}\times 10^{3}, -1.4562\Biggr).
		\end{align*}
		
		Following Proposition \ref{prop:mu_est} the scalars \(\check{\mu}_{pq}\), \((p,q)\in E(\P)\) are computed as: \(\check{\mu}_{pq} = \lambda_{\min}(P_{q}P_{p}^{-1})\). We have
		\begin{align*}
			\check{\mu}_{12} &= 0.2661, & \check{\mu}_{13} &= 5.6395, &
			\check{\mu}_{21} &= 0.6446, & \check{\mu}_{24} &= 1.3875\times 10^{3},\\
			\check{\mu}_{31} &= 0.0173, & \check{\mu}_{34} &= 87.0252,&
			\check{\mu}_{42} &= 1.4133\times 10^{-4}, & \check{\mu}_{43} &= 0.0070.
		\end{align*}
		\\
		{\it B. Switching signal \(\sigma\)}: Fix \(t > 0\). Let a switching signal \(\sigma\) satisfy
		\begin{itemize}[label = \(\circ\), leftmargin = *]
			\item \(\Ntsigma(t) = \frac{t}{20}\), \(\Ntsigma_{pq}(t) = \frac{\Ntsigma(t)}{8}\), \((p,q)\in E(\P)\),
			\item \(\T_{p}(t) = \frac{t}{8}\), \(p\in\{1,2\}\) and \(\T_{p}(t) = \frac{3t}{8}\), \(p\in\{3,4\}\).
		\end{itemize}
		Consequently,
		\begin{itemize}[label = \(\circ\), leftmargin = *]
			\item \(\nu(t) = \frac{1}{20}\),
			\item \(\rho_{pq}(t) = \frac{1}{8}\), \((p,q)\in E(\P)\),
			\item \(\eta_{p}(t) = \frac{1}{8}\), \(p\in\{1,2\}\) and \(\eta_{p}(t) = \frac{3}{8}\), \(p\in\{3,4\}\).
		\end{itemize}
		\vspace*{0.3cm}
		{\it C. Verification of condition \eqref{e:instab_condn}}: We have
		\begin{align*}
			\varliminf_{t\to+\infty}\Biggl(\nu(t)\sum_{(p,q)\in E(\P)}(\ln\check{\mu}_{pq})\rho_{pq}(t) &- \sum_{p\in\P_{S}}\abs{\check{\lambda}_{p}}\eta_{p}(t)
			+\sum_{p\in\P_{U}}\abs{\check{\lambda}_{p}}\eta_{p}(t)\Biggr) = 0.1064 > 0.
		\end{align*}
		
		In Figures \ref{fig:ex2_sw_plot1} and \ref{fig:ex2_xt_plot1} we illustrate \((\sigma(t))_{t\geq 0}\) and \((\norm{x(t)})_{t\geq 0}\) till \(t = 200\) units of time, respectively. For plotting \((\norm{x(t)})_{t\geq 0}\), ten different initial conditions are chosen from \([-10,10]^{2}\) uniformly at random. Divergence of \(\norm{x(t)}\) is observed in each case.
		\begin{figure}[htbp]
			\includegraphics[scale = 0.45]{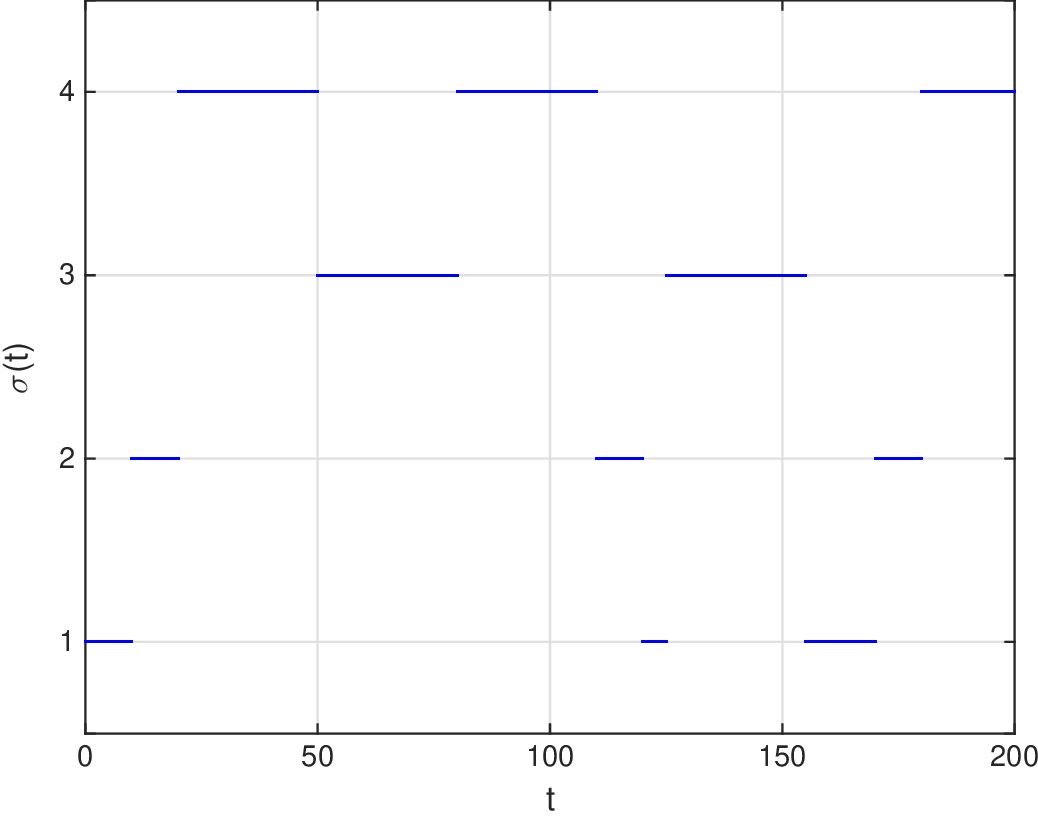}
			\caption{Plot of \((\sigma(t))_{t\geq 0}\) for Example \ref{ex:exmpl1}} \label{fig:ex2_sw_plot1}
		\end{figure}
		\begin{figure}[htbp]
			\includegraphics[scale = 0.45]{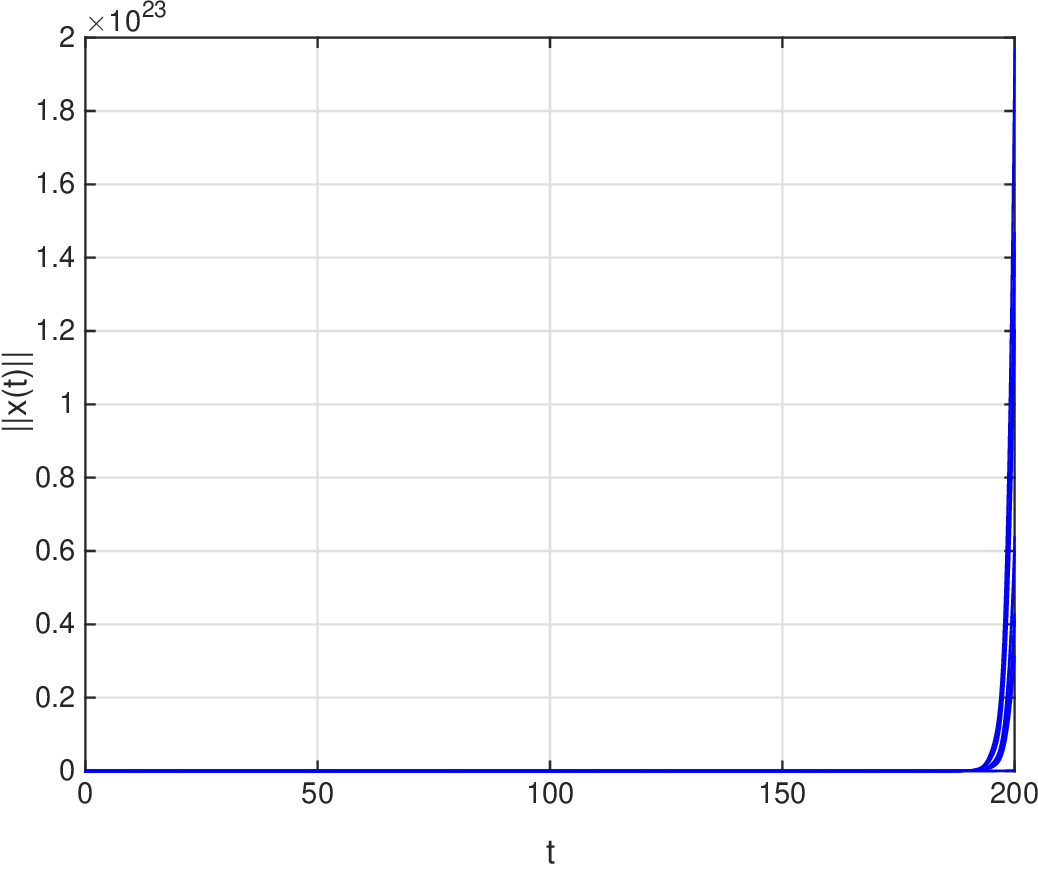}
			\caption{Plot of \((\norm{x(t))})_{t\geq 0}\) for Example \ref{ex:exmpl1}} \label{fig:ex2_xt_plot1}
		\end{figure}
	}
	\end{example}
\subsection{Gap between stability and instability}
\label{ss:gap}
	Given a family of systems \eqref{e:family} and a set of admissible transitions \(E(\P)\), Theorems \ref{t:stab_thm} and \ref{t:instab_thm} characterize classes of stabilizing and destabilizing switching signals \(\Sss\) and \(\Ssi\), respectively. Both the characterizations involve asymptotic properties of switching signals, are derived by employing multiple Lyapunov-like functions, and are sufficient in nature. It is immediate that the elements of \(\Ssi\) are contained in \(\Sss^{c}\). Indeed, fix a switching signal \(\sigma\in\Ssi\). We have \eqref{e:key_ineq1} and \eqref{e:key_ineq3} with \(\check{\lambda}_{p}\geq \hat{\lambda}_{p}\) for \(p\in\P_{S}\) and \(\abs{\check{\lambda}_{p}}\leq\abs{\hat{\lambda}_{p}}\) for \(p\in\P_{U}\). Also, \eqref{e:key_ineq2} and \eqref{e:key_ineq4} hold with \(\check{\mu}_{pq}\leq\hat{\mu}_{pq}\) for each \((p,q)\in E(\P)\). Consequently,
	\begin{align*}
		\nu(t)\sum_{(p,q)\in E(\P)}(\ln\check{\mu}_{pq})\rho_{pq}(t) &- \sum_{p\in\P_{S}}\abs{\check{\lambda}_{p}}\eta_{p}(t)
			+\sum_{p\in\P_{U}}\abs{\check{\lambda}_{p}}\eta_{p}(t)
	\end{align*}
	is at most equal to
	\[
		\nu(t)\sum_{(p,q)\in E(\P)}(\ln\hat{\mu}_{pq})\rho_{pq}(t) - \sum_{p\in\P_{S}}\abs{\hat{\lambda}_{p}}\eta_{p}(t)
		+ \sum_{p\in\P_{U}}\abs{\hat{\lambda}_{p}}\eta_{p}(t)
	\]
	for all \(t > 0\). Since \(\sigma\in\Ssi\),
	\begin{align*}
	\varlimsup_{t\to+\infty}\Biggl(\nu(t)\sum_{(p,q)\in E(\P)}(\ln\hat{\mu}_{pq})\rho_{pq}(t) &- \sum_{p\in\P_{S}}\abs{\hat{\lambda}_{p}}\eta_{p}(t)
			+ \sum_{p\in\P_{U}}\abs{\hat{\lambda}_{p}}\eta_{p}(t) \Biggr) \nless 0.
	\end{align*}
	
	The following proposition asserts that \(\Ssi\) is a strict subset of \(\Sss^{c}\).
	\begin{prop}
	\label{prop:gap}
	\rm{
		Consider a family of systems \eqref{e:family}. Then an element \(\sigma\) of the set \(\Sss^{c}\) is not necessarily an element of the set \(\Ssi\).
	}
	\end{prop}
	We prove Proposition \ref{prop:gap} with an example \(\sigma\) that satisfies neither condition \eqref{e:stab_condn} nor condition \eqref{e:instab_condn}, see \S\ref{s:all_proofs2}.
	
	Proposition \ref{prop:gap} implies that there is a gap between the characterization of stabilizing and destabilizing switching signals based solely on the asymptotic behaviour of these signals. We use the term ``gap'' in the following sense: (de)stabilizing properties of the elements of \(\Sss^{c}\cap\Ssi^{c}\) cannot be determined from conditions \eqref{e:stab_condn} and \eqref{e:instab_condn}. In Example \ref{ex:counter_ex1} we presented a stabilizing switching signal \(\sigma\) that belongs to \(\Sss^{c}\cap\Ssi^{c}\).\footnote{A proof for \(\sigma\notin\Ssi^{c}\) is given in our proof of Proposition \ref{prop:gap}.} Below we provide an example for a destabilizing switching signal \(\sigma'\) belonging to the set \(\Sss^{c}\cap\Ssi^{c}\).
		\begin{example}
		\label{ex:exmpl2}
		\rm{
			Consider a family of systems \eqref{e:family} with \(\P = \{1,2\}\), where \(A_{1} = \pmat{-0.3509 & -0.2683\\-0.3523 & -0.5491}\) and \(A_{2} = \pmat{0.6294 & 0.8116\\-0.7460 & 0.8268}\). We have \(\P_{S} = \{1\}\) and \(\P_{U} = \{2\}\). Let \(E(\P) = \{(1,2),(2,1)\}\).
			
			By choosing \(Q_{1} = Q_{2} = I_{2\times 2}\), \(\varepsilon_{2} = 0.7278\), we obtain the following estimates of \(P_{p}\), \(\hat{\lambda}_{p}\), \(\check{\lambda}_{p}\), \(p\in\P\) and \(\hat{\mu}_{pq}\), \(\check{\mu}_{pq}\), \((p,q)\in E(\P)\):
			\begin{align*}
				P_{1} &= \pmat{2.9643 & -1.5333\\-1.5333 & 1.6598}, &\quad P_{2} &= \pmat{4.8542 & -0.6419\\-0.6419 & 5.2809}\times 10^3,
			\end{align*}
			\begin{align*}
				\check{\lambda}_{1} &= 1.5486, &\quad \hat{\lambda}_{1} &= 0.2514,\\
				\check{\lambda}_{2} &= -1.4562, &\quad \hat{\lambda}_{2} &= -1.4562,\\
				\check{\mu}_{12} &= 1.3875\times 10^3, &\quad \hat{\mu}_{12} &= 7.0755\times 10^3,\\
				\check{\mu}_{21} &= 1.4133\times 10^{-4}, &\quad \hat{\mu}_{21} &= 7.2070\times 10^{-4}.
			\end{align*}
			
			Consider a switching signal \(\sigma'\) that satisfies
			\begin{itemize}[label = \(\circ\), leftmargin = *]
				\item \(\Ntsigma(t) = \lfloor\frac{t}{10}\rfloor\),
				\item \(\Ntsigma_{pq}(t) = \frac{\Ntsigma(t)}{2}\), \((p,q)\in E(\P)\),
				\item \(\T_{p}(t) = \frac{t}{2}\), \(p\in\P\).
			\end{itemize}
			We have
			\begin{itemize}[label = \(\circ\), leftmargin = *]
				\item \(\nu(t) = \frac{1}{10}\),
				\item \(\rho_{pq}(t) = \frac{1}{2}\), \((p,q)\in E(\P)\),
				\item \(\eta_{p}(t) = \frac{1}{2}\), \(p\in\P\).
			\end{itemize}
			
			Consequently,
			\begin{align*}
				\varlimsup_{t\to+\infty}\Biggl(\nu(t)&\sum_{(p,q)\in E(\P)}(\ln\hat{\mu}_{pq})\rho_{pq}(t) - \sum_{p\in\P_{S}}\abs{\hat{\lambda}_{p}}\eta_{p}(t)
			+ \sum_{p\in\P_{U}}\abs{\hat{\lambda}_{p}}\eta_{p}(t) \Biggr) = 0.6839 > 0,\\
				\varliminf_{t\to+\infty}\Biggl(\nu(t)&\sum_{(p,q)\in E(\P)}(\ln\check{\mu}_{pq})\rho_{pq}(t) - \sum_{p\in\P_{S}}\abs{\check{\lambda}_{p}}\eta_{p}(t)
			+ \sum_{p\in\P_{U}}\abs{\check{\lambda}_{p}}\eta_{p}(t) \Biggr) = -0.1277 < 0,
			\end{align*}
			and \(\sigma\in\Sss^{c}\cap\Ssi^{c}\).
			
			We observe that \(\sigma\) is destabilizing. In Figures \ref{fig:ex3_swplot1} and \ref{fig:ex3_xplot} we demonstrate \((\sigma(t))_{t\geq 0}\) and \((\norm{x(t)})_{t\geq 0}\) till \(t=100\) units of time, respectively. The initial condition \(x(0)\) is chosen as: \(\pmat{9.0044\\-9.3111}\).
			\begin{figure}[htbp]
			\includegraphics[scale = 0.45]{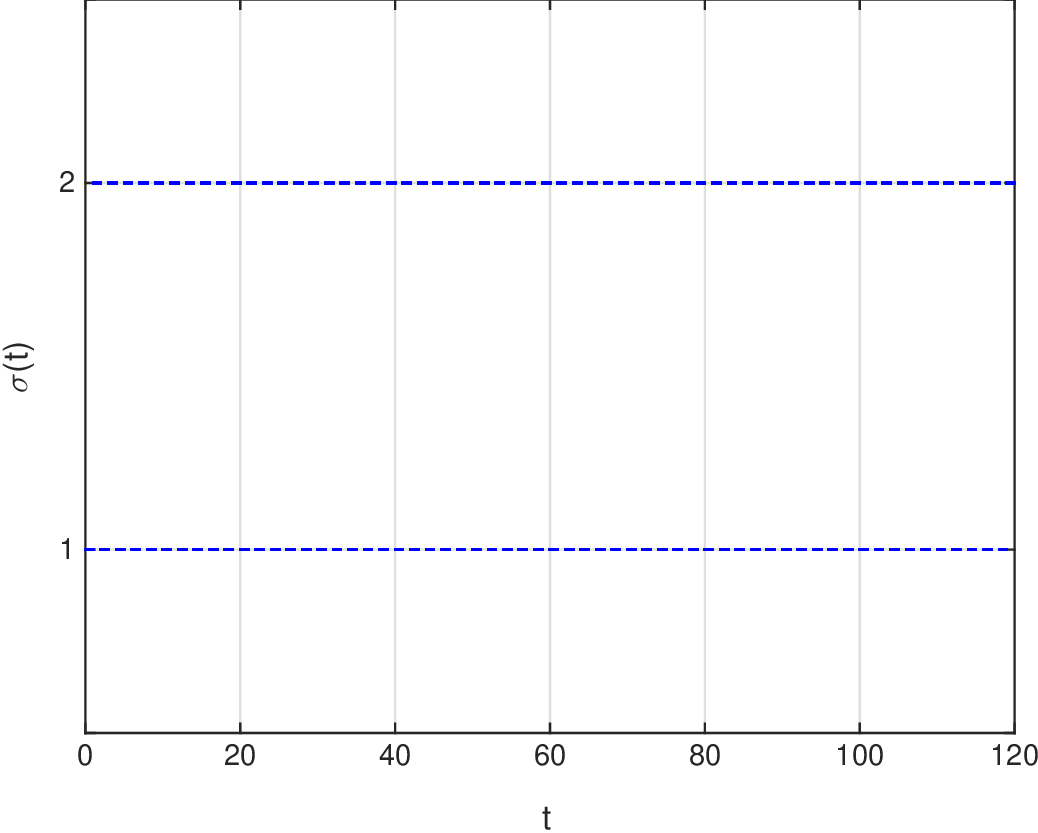}
			\caption{Plot of \((\sigma(t))_{t\geq 0}\) for Example \ref{ex:exmpl1}} \label{fig:ex3_swplot1}
		\end{figure}
		\begin{figure}[htbp]
			\includegraphics[scale = 0.45]{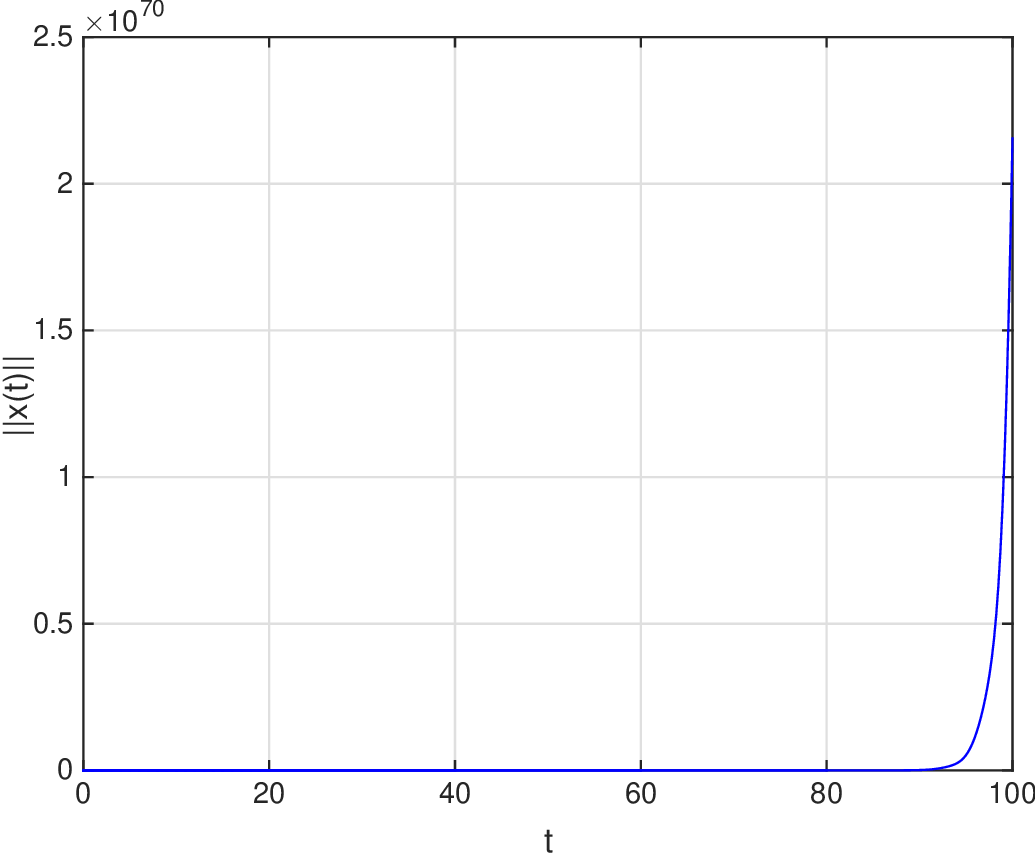}
			\caption{Plot of \((\norm{x(t))})_{t\geq 0}\) for Example \ref{ex:exmpl2}} \label{fig:ex3_xplot}
		\end{figure}
		}
		\end{example}
	\begin{remark}
	\label{rem:gap}
	\rm{
		Recall that for asymptotic characterization of stabilizing switching signals, convergence of \(\norm{x(t)}\) for all \(x(0)\) under the ``worst case'' switching is ensured, while for asymptotic characterization of destabilizing switching signals divergence of \(\norm{x(t)}\) for all \(x(0)\) under the ``best case'' switching is ensured. It is, therefore, immediate that conditions \eqref{e:stab_condn} and \eqref{e:instab_condn} are not sufficient to conclude (in)stability of the switched system \eqref{e:swsys} under the elements of the set \(\Sss^{c}\cap\Ssi^{c}\). Examples \ref{ex:exmpl1} and \ref{ex:exmpl2} hint that given a family of systems \eqref{e:family}, (de)stabilizing properties of a switching signal \(\sigma\in\Sss^{c}\cap\Ssi^{c}\) depend on the properties of the subsystems. We considered a family with all systems asymptotically stable to demonstrate a stabilizing \(\sigma\in\Sss^{c}\cap\Ssi^{c}\), and a family containing both asymptotically stable and unstable systems to demonstrate a destabilizing switching signal \(\sigma'\in\Sss^{c}\cap\Ssi^{c}\). Facts \ref{fact:key1}-\ref{fact:key4} are not sufficient to address (in)stability of \eqref{e:swsys} under a \(\sigma\in\Sss^{c}\cap\Ssi^{c}\), and additional analysis tools are needed.
		}
		\end{remark}
		
		Below we summarize the classes of switching signals discussed in this paper:
		\begin{itemize}[label = \(\circ\),leftmargin = *]
			\item \(\Ss\): set of all switching signals,
			\item \(\Sss\subset\Ss\): set of switching signals that satisfy condition \eqref{e:stab_condn} (stabilizing),
			\item \(\Ssi\subset\Sss^{c}\): set of switching signals that satisfy condition \eqref{e:instab_condn} (destabilizing),
			\item \(\Sss^{c}\cap\Ssi^{c}\): set of switching signals that satisfy neither condition \eqref{e:stab_condn} nor condition \eqref{e:instab_condn} (gap).
		\end{itemize}
\section{Concluding remarks}
\label{s:concln}
	In this paper we proposed a class of destabilizing switching signals for continuous-time switched linear systems. Our characterization of instability is based on the asymptotic behaviour of frequency of switching, frequency of transitions between subsystems, and fraction of activation of subsystems. Asymptotic properties of switching signals were employed to characterize stability of a switched linear system earlier in \cite{pqr}. We showed that our class of destabilizing switching signals is a strict subset of the class of switching signals that does not satisfy the stability condition of \cite{pqr}, and thereby identified a gap between asymptotic characterization of stabilizing and destabilizing switching signals. An interesting future direction is about determining (de)stabilizing properties of the switching signals that satisfy neither the stability condition of \cite{pqr} nor the instability condition proposed in this paper. As discussed in Remark \ref{rem:gap}, this problem cannot be addressed with multiple Lyapunov-like functions based analysis in its standard form, and additional tools are required.
\section{Proofs of Fact \ref{fact:key3} and Proposition \ref{prop:mu_est}}
\label{s:all_proofs1}
	
	\begin{proof}[Proof of Fact \ref{fact:key3}]
		We begin with asymptotically stable subsystems \(p\in\P_{S}\). Let \(\R^{d}\ni\xi\mapsto V_{p}(\xi):= <P_{p}\xi,\xi>\), where \(P_{p}\in\R^{d\times d}\) is a symmetric and positive definite solution to the Lyapunov equation
		\begin{align}
		\label{e:Lyap_eqn}
			A_{p}^\top P_{p} + P_{p}A_{p} = -Q_{p}
		\end{align}
		for some pre-selected symmetric and positive definite matrix \(Q_{p}\in\R^{d\times d}\) \cite[Corollary 11.9.1]{Bernstein}. Recall that \cite[Lemma 8.4.3]{Bernstein} any symmetric and positive definite matrix \(Z\in\R^{d\times d}\) satisfies
		\[
			\lambda_{\min}(Z)\norm{z}^{2} \leq z^\top Zz \leq \lambda_{\max}(Z)\norm{z}^{2}.
		\]
		Consequently, for all \(\xi\in\R^{d}\), we have
		\[
			\xi^\top Q_{p}\xi \geq -\frac{\lambda_{\max}(Q_{p})}{\lambda_{\min}(P_{p})}\xi^\top P_{p}\xi.
		\]
		Let \(\check{\lambda}_{p} = \frac{\lambda_{\max}(Q_{p})}{\lambda_{\min}(P_{p})}\). We have
		\[
			\frac{d}{dt}V_{p}(\gamma_{p}(t)) \geq -\check{\lambda}_{p} V_{p}(\gamma_{p}(t))
		\]
		leading to \eqref{e:key_ineq3} with \(\check{\lambda}_{p} > 0\).
		
		We now move on to unstable subsystems \(p\in\P_{U}\). There exist \(\varepsilon_{p} > 0\) such that \(A_{p} - \varepsilon_{p}I\) is asymptotically stable. Select the Lyapunov-like function \(\R^{d}\ni\xi\mapsto V_{p}(\xi) := <P_{p}\xi,\xi>\), where \(P_{p}\in\R^{d\times d}\) is a symmetric and positive definite solution to the Lyapunov equation \eqref{e:Lyap_eqn} with \(A_{p} = A_{p} - \varepsilon_{p}I\) and a pre-selected symmetric and positive definite matrix \(Q_{p}\in\R^{d\times d}\). Following the set of arguments for asymptotically stable subsystems, we have that
		\[
			\frac{d}{dt}V_{p}(\gamma_{p}(t)) \geq -\check{\lambda}_{p}V_{p}(\gamma_{p}(t)),
		\]
		where \(\displaystyle{\check{\lambda}_{p} = -\biggl(2\varepsilon_{p} - \frac{\lambda_{\max}(Q_{p})}{\lambda_{\min}(P_{p})}\biggr)}\). Notice that \(\varepsilon_{p}\) is any scalar strictly bigger than \(Re(\lambda_{\max}(A_{p}))\). One needs to choose \(\varepsilon_{p}\) and \(Q_{p}\) such that the term \(2\varepsilon_{p} - \frac{\lambda_{\max}(Q_{p})}{\lambda_{\min}(P_{p})} \geq 0\). Consequently, \eqref{e:key_ineq3} follows with \(\check{\lambda}_{p}\leq 0\).
	\end{proof}
	
	\begin{proof}[Proof of Proposition \ref{prop:mu_est}]
		Recall that by definition of \(V_{p}\), \(p\in\P\) in \eqref{e:Lyap_func-defn}, each \(P_{p}\), \(p\in\P\) is symmetric and positive definite. Hence, \(P_{p}^{-1}\), \(p\in\P\) exist. Also, \(P_{q}P_{p}^{-1}\) is similar to the matrix \(P_{p}^{-1/2}(P_{q}P_{p}^{-1})P_{p}^{1/2}\), and the latter is symmetric and positive definite. Since the spectrum of a matrix is invariant under similarity transformations, the eigenvalues of \(P_{q}P_{p}^{-1}\) are the same as the eigenvalues of \(P_{p}^{-1/2}P_{q}P_{p}^{-1/2}\), and consequently, the eigenvalues of \(P_{q}P_{p}^{-1}\) are real numbers.
		
		Now,
		\begin{align*}
			\inf_{\R^{d}\ni\xi\neq 0}\frac{V_{q}(\xi)}{V_{p}(\xi)} = \inf_{\R^{d}\ni\xi\neq 0}\frac{<P_{q}\xi,\xi>}{<P_{p}\xi,\xi>}
			= \inf_{\R^{d}\ni\xi\neq 0}\frac{<P_{q}\xi,\xi>}{<P_{p}^{1/2}\xi,P_{p}^{1/2}\xi>}.
		\end{align*}
		Let \(y = P_{p}^{1/2}\xi\). Then the right-hand side of the above equality is same as
		\begin{align*}
			\inf_{\R^{d}\ni y\neq 0}\frac{<P_{q}(P_{p}^{-1/2}y),P_{p}^{-1/2}y>}{<y,y>} &= \inf_{\R^{d}\ni y \neq 0}\frac{<P_{p}^{-1/2}P_{q}P_{p}^{-1/2}y,y>}{<y,y>}\\
			&= \lambda_{\min}(P_{p}^{-1/2}P_{q}P_{p}^{-1/2})\\
			&= \lambda_{\min}(P_{q}P_{p}^{-1}).
		\end{align*}
		Since \(V_{q}(\xi)\geq\check{\mu}_{pq}V_{p}(\xi)\) for all \(\xi\in\R^{d}\), the constant \(\check{\mu}_{pq}\) satisfies \eqref{e:mu_est}.
	\end{proof}
\section{Proofs of Theorem \ref{t:instab_thm} and Proposition \ref{prop:gap}}
\label{s:all_proofs2}
	
	\begin{proof}[Proof of Theorem \ref{t:instab_thm}]
		Recall that \(0=:\tau_{0}<\tau_{1}<\cdots<\tau_{\Ntsigma(t)}\) are the switching instants before (and including) \(t > 0\).
		
		In view of \eqref{e:key_ineq1}, we have
			\begin{align}
			\label{e:pf1_step1}
				V_{\sigma(t)}(x(t)) \geq \exp\bigl(-\check{\lambda}_{\sigma(\tau_{\Ntsigma(t)})}(t-\tau_{\Ntsigma(t)})\bigr)V_{\sigma(t)}(x(\tau_{\Ntsigma(t)})).
			\end{align}
		A straightforward iteration of \eqref{e:pf1_step1} applying \eqref{e:key_ineq3} and \eqref{e:key_ineq4}, we obtain
			\begin{align}
			\label{e:pf1_step2}
				V_{\sigma(t)}(x(t)) \geq \exp\Biggl(&-\sum_{\substack{i=0\\\tau_{\Ntsigma(t)+1}:=t}}^{\Ntsigma(t)}\check{\lambda}_{\sigma(\tau_{i})}(\tau_{i+1}-\tau_{i})\times\prod_{i=0}^{\Ntsigma(t)-1}\check{\mu}_{\sigma(\tau_{i})\sigma(\tau_{i+1})}\Biggl)V_{\sigma(0)}(x(0)).
			\end{align}
			
		Now,
			\begin{align}
			\label{e:pf1_step3}
				\prod_{i=0}^{\Ntsigma(t)-1}\check{\mu}_{\sigma(\tau_{i})\sigma(\tau_{i+1})} = \exp\Biggl(\sum_{i=0}^{\Ntsigma(t)-1}\ln\check{\mu}_{\sigma(\tau_{i})\sigma(\tau_{i+1})}\Biggr)
				&=\exp\Biggl(\sum_{p\in\P}\sum_{i=0}^{\Ntsigma(t)-1}\sum_{\substack{(p,q)\in E(\P),\\\sigma(\tau_{i}) = p,\\\sigma(\tau_{i+1}) = q}}\ln\check{\mu}_{pq}\Biggr)\nonumber\\
				&=\exp\Biggl(\sum_{(p,q)\in E(\P)}\bigl(\ln\check{\mu}_{pq}\bigr)\Ntsigma_{pq}(t)\Biggr)\nonumber\\
				&=\exp\Biggl(\Ntsigma(t)\sum_{(p,q)\in E(\P)}\bigl(\ln\check{\mu}_{pq}\bigr)\rho_{pq}(t)\Biggr),
			\end{align}
			where \(\rho_{pq}(t)\), \((p,q)\in E(\P)\) is as defined in \eqref{e:trans_freq}.\\
			Also,
			\begin{align*}
				&\exp\Biggl(-\sum_{\substack{i=0\\\tau_{\Ntsigma(t)+1 : = t}}}^{\Ntsigma(t)}\check{\lambda}_{\sigma(\tau_{i})}(\tau_{i+1}-\tau_{i})\Biggr)= \exp\Biggl(-\sum_{i=0}^{\Ntsigma(t)-1}\Biggl(\sum_{p\in\P}\mathrm{1}(\sigma(\tau_{i}) = p)\check{\lambda}_{p}(\tau_{i+1}-\tau_{i})\Biggr)\Biggr).
			\end{align*}
			Separating out the asymptotically stable and unstable subsystems, we see that the right-hand side of the above expression is equal to
			\begin{align*}
				\exp\Biggl(-\Biggl(&\sum_{p\in\P_{S}}\check{\lambda}_{p}\sum_{\substack{i:\sigma(\tau_{i}) = p\\i=0,1,\ldots,\Ntsigma(t),\\\tau_{\Ntsigma(t)+1} := t}}(\tau_{i+1}-\tau_{i})+ \sum_{p\in\P_{U}}\check{\lambda}_{p}\sum_{\substack{i:\sigma(\tau_{i}) = p\\i=0,1,\ldots,\Ntsigma(t),\\\tau_{\Ntsigma(t)+1} := t}}(\tau_{i+1}-\tau_{i})\Biggr)\Biggr).
			\end{align*}
			Recall that \(\check{\lambda}_{p} > 0\) for \(p\in\P_{S}\) and \(\check{\lambda}_{p}\leq 0\) for \(p\in\P_{U}\). Consequently, the above expression is equal to
			\begin{align}
			\label{e:pf1_step4}
				&\exp\Biggl(-\sum_{p\in\P_{S}}\abs{\check{\lambda}_{p}}\sum_{\substack{i:\sigma(\tau_{i}) = p\\i=0,1,\ldots,\Ntsigma(t),\\\tau_{\Ntsigma(t)+1} := t}}(\tau_{i+1}-\tau_{i})+\sum_{p\in\P_{U}}\abs{\check{\lambda}_{p}}\sum_{\substack{i:\sigma(\tau_{i}) = p\\i=0,1,\ldots,\Ntsigma(t),\\\tau_{\Ntsigma(t)+1} := t}}(\tau_{i+1}-\tau_{i})\Biggr)\nonumber\\
				=&\exp\Biggl(-\sum_{p\in\P_{S}}\abs{\check{\lambda}_{p}}\T_{p}(t)+\sum_{p\in\P_{U}}\abs{\check{\lambda}_{p}}\T_{p}(t)\Biggr),
			\end{align}
			where \(\T_{p}(t)\), \(p\in\P\) is as defined in \S\ref{ss:stab_sw}.
			
			Substituting \eqref{e:pf1_step3} and \eqref{e:pf1_step4} into \eqref{e:pf1_step2}, we have
			\begin{align}
			\label{e:pf1_step5}
				V_{\sigma(t)}(x(t))\geq \exp\Bigl(\psi(t)\Bigr)V_{\sigma(0)}(x(0)),
			\end{align}
			where, for \(t > 0\), the function \(\psi\) is defined as
			\begin{align}
			\label{e:pf1_step6}
				\psi(t) := \Ntsigma(t)\sum_{(p,q)\in E(\P)}(\ln\check{\mu})\rho_{pq}(t) &-\sum_{p\in\P_{S}}\abs{\check{\lambda}_{p}}\T_{p}(t)
				+\sum_{p\in\P_{U}}\abs{\check{\lambda}_{p}}\T_{p}(t).
			\end{align}
			
			For \(t > 0\), the right-hand side above is equal to
			\begin{align}
			\label{e:pf1_step7}
				t\Biggl(\frac{\Ntsigma(t)}{t}\sum_{(p,q)\in E(\P)}(\ln\check{\mu}_{pq})\rho_{pq}(t) &- \sum_{p\in\P_{S}}\abs{\check{\lambda}_{p}}\frac{\T_{p}}{t}
				+\sum_{p\in\P_{U}}\abs{\check{\lambda}_{p}}\frac{\T_{p}(t)}{t}\Biggr)\nonumber\\
				t\Biggl(\nu(t)\sum_{(p,q)\in E(\P)}(\ln\check{\mu}_{pq})\rho_{pq}(t) &- \sum_{p\in\P_{S}}\abs{\check{\lambda}_{p}}\eta_{p}(t)
				+\sum_{p\in\P_{U}}\abs{\check{\lambda}_{p}}\eta_{p}(t)\Biggr),
			\end{align}
			where \(\nu(t)\) and \(\eta_{p}(t)\), \(p\in\P\) are as defined in \eqref{e:sw_freq} and \eqref{e:frac_activ}, respectively.
			
			From the definition of \(V_{p}\), \(p\in\P\) in \eqref{e:Lyap_func-defn}, we have
			\begin{align}
			\label{e:Lyap_key}
				\underline{\alpha}(\norm{\xi})\leq V_{p}(\xi) \leq\overline{\alpha}(\norm{\xi})\:\:\text{for all}\:\xi\in\R^{d},
			\end{align}
			where \(\displaystyle{\underline{\alpha}(r) := \min_{p\in\P}\Bigl(\lambda_{\min}(P_{p})\Bigr)r^{2}}\) and \(\displaystyle{\overline{\alpha}(r) := \max_{p\in\P}\Bigl(\lambda_{\min}(P_{p})\Bigr)r^{2}}\). It follows that
			\begin{align}
			\label{e:pf1_step8}
				\overline{\alpha}(\norm{x(t)})\geq\exp(\psi(t))\underline{\alpha}(\norm{x(0)}).
			\end{align}
			
			Armed with \eqref{e:pf1_step8}, to verify instability of the switched linear system \eqref{e:swsys} (according to Definition \ref{d:uns}), we find conditions such that
			\begin{align}
			\label{e:pf1_step9}
				\lim_{t\to+\infty}\exp(\psi(t)) = +\infty.
			\end{align}
			The following condition is sufficient to guarantee \eqref{e:pf1_step9}.
			\begin{align*}
				\varliminf\Biggl(\nu(t)\sum_{(p,q)\in E(\P)}(\ln\check{\mu}_{pq})\rho_{pq}(t) &- \sum_{p\in\P_{S}}\abs{\check{\lambda}_{p}}\eta_{p}(t)
				+\sum_{p\in\P_{U}}\abs{\check{\lambda}_{p}}\eta_{p}(t)\Biggr) > 0.
			\end{align*}
			This completes our proof of Theorem \ref{t:instab_thm}.
	\end{proof}
	
	\begin{remark}
	\label{rem:stab-instab_diff}
	\rm{
		From Facts \ref{fact:key1} and \ref{fact:key2}, we obtain
		\[
			V_{\sigma(t)}(x(t))\leq\exp(\phi(t))V_{\sigma(0)}(x(0)),
		\]
		where the function \(\phi(t)\) involves the scalars \(\hat{\lambda}_{p}\), \(p\in\P\) and \(\hat{\mu}_{pq}\), \((p,q)\in E(\P)\). In view of \eqref{e:Lyap_key}, we obtain
		\[
			\underline{\alpha}(\norm{x(t)})\leq\exp(\phi(t))(\norm{x(0)}).
		\]
		The function \(\phi\) is analyzed for Lyapunov stability and global asymptotic convergence of the switched linear system \eqref{e:swsys}. In contrast we analyze the function \(\psi\) for asymptotic divergence of \(\norm{x(t)}\).
	}
	\end{remark}
	
	\begin{proof}[Proof of Proposition \ref{prop:gap}]
		Consider the setting of Example \ref{ex:counter_ex1}. We have \(A_{1} = \pmat{-0.2 & -0.4\\3 & -0.2}\) and \(A_{2} = \pmat{-0.2 & -3\\0.4 & -0.2}\).
		
		Select \(V_{p}(\xi) = <P_{p}\xi,\xi>\), \(p\in\P\), where \(P_{1} = \pmat{10.3629 & 0.5242\\0.5242 & 1.4516}\), \(P_{2} = \pmat{1.4516 & -0.5242\\-0.5242 & 10.3629}\) are obtained by solving the Lyapunov equation \(A_{p}^\top P_{p}+P_{p}A_{p} = -Q_{p}\) with \(Q_{p} = I_{d\times d}\), \(p\in\P\).
		
		Let a switching signal \(\sigma\) satisfy
		\begin{itemize}[label = \(\circ\), leftmargin = *]
			\item \(\Ntsigma(t) = \lfloor\frac{t}{10}\rfloor\),
			\item \(\Ntsigma_{pq}(t) = \frac{\Ntsigma(t)}{2}\), \((p,q)\in E(\P)\),
			\item \(\T_{p}(t) = \frac{t}{2}\), \(p\in\P\).
		\end{itemize}
		Clearly,
		\begin{itemize}[label = \(\circ\), leftmargin = *]
			\item \(\nu(t) = \frac{1}{10}\),
			\item \(\rho_{pq}(t) = \frac{1}{2}\), \((p,q)\in E(\P)\),
			\item \(\eta_{p}(t) = \frac{1}{2}\), \(p\in\P\).
		\end{itemize}
		
		In Example \ref{ex:counter_ex1} we showed that the above \(\sigma\) does not satisfy condition \eqref{e:stab_condn}, and hence \(\sigma\) is an element of \(\Sss^{c}\). We will now show that \(\sigma\) is not an element of \(\Ssi\) either.
		
		Corresponding to the chosen pairs \((P_{1},Q_{1})\) and \((P_{2},Q_{2})\), we obtain the following estimates of the scalars \(\check{\lambda}_{p}\), \(p\in\P\) and \(\check{\mu}_{pq}\), \((p,q)\in E(\P)\):\footnote{We employ \(\check{\lambda}_{p} = \frac{\lambda_{\max}(Q_{p})}{\lambda_{\min}(P_{p})}\), \(p\in\P\) and \(\check{\mu}_{pq} = \lambda_{\min}(P_{q}P_{p}^{-1})\), \((p,q)\in E(\P)\).}
		\[
			\check{\lambda}_{1} = \check{\lambda}_{2} = 0.7038,\quad\check{\mu}_{12} = \check{\mu}_{21} = 0.1367.
		\]
		Now,
		\begin{align*}
			\varliminf_{t\to+\infty}\Biggl(\nu(t)&\sum_{(p,q)\in E(\P)}(\ln\check{\mu}_{pq})\rho_{pq}(t) - \sum_{p\in\P_{S}}\abs{\check{\lambda}_{p}}\eta_{p}(t)
			+\sum_{p\in\P_{U}}\abs{\check{\lambda}_{p}}\eta_{p}(t)\Biggr) = -0.9028 < 0,
		\end{align*}
		and hence \(\sigma\notin\Ssi\).
	\end{proof}


\end{document}